\documentclass[10pt]{article}

\usepackage[letterpaper,textwidth=5.5in,textheight=9in,top=1in,headheight=12pt,headsep=0.25in,centering]{geometry}
\usepackage{fancyhdr}
\usepackage[explicit]{titlesec}
\usepackage[T1]{fontenc}
\usepackage[utf8]{inputenc}
\usepackage{mathptmx}
\usepackage[scaled=0.92]{helvet}
\usepackage{microtype}

\usepackage{amsmath,amssymb,amsthm}
\usepackage{bm}

\usepackage{booktabs}
\usepackage{makecell}
\usepackage{array}
\usepackage{enumitem}
\usepackage[font=small,labelfont=bf,skip=5pt]{caption}

\usepackage{xcolor}
\usepackage{graphicx}
\usepackage{tikz}
\usetikzlibrary{arrows.meta,positioning,calc,decorations.pathreplacing}
\usepackage{pgfplots}
\pgfplotsset{compat=1.16}
\usepgfplotslibrary{groupplots}
\usepackage{algorithm}
\usepackage[noend]{algpseudocode}

\usepackage[numbers,sort&compress]{natbib}
\usepackage[colorlinks=true,linkcolor=linkblue,citecolor=linkblue,urlcolor=linkblue]{hyperref}
\usepackage[capitalise,nameinlink]{cleveref}

\definecolor{probe}{HTML}{2F6FAF}
\definecolor{detect}{HTML}{B35C10}
\definecolor{linkblue}{HTML}{1F4E8C}
\definecolor{softgray}{HTML}{F4F5F7}
\definecolor{midgray}{HTML}{8A8F98}
\definecolor{okgreen}{HTML}{2E7D4F}
\definecolor{badred}{HTML}{B03A2E}

\theoremstyle{plain}
\newtheorem{proposition}{Proposition}

\theoremstyle{definition}
\newtheorem{definition}{Definition}

\newcommand{\pwrong}{p_{\mathrm{wrong}}}
\newcommand{\Nflag}{N_{\mathrm{flag}}}
\newcommand{\pistar}{\pi^{\ast}}
\newcommand{\cmark}{\textcolor{okgreen}{\ensuremath{\checkmark}}}
\newcommand{\xmark}{\textcolor{badred}{\ensuremath{\times}}}
\newcommand{\pmark}{\textcolor{midgray}{\ensuremath{\sim}}}
\newcommand{\type}[1]{\textsc{#1}}

\setlist{itemsep=1pt,topsep=3pt}

\fancypagestyle{plain}{\fancyhf{}\fancyfoot[C]{\thepage}}

\titleformat{\section}{\normalfont\large\bfseries}{\thesection}{1em}{#1}
\titleformat{name=\section,numberless}{\normalfont\large\bfseries}{}{0pt}{#1}
\titleformat{\subsection}{\normalfont\normalsize\bfseries}{\thesubsection}{1em}{#1}
\titleformat{\paragraph}[runin]{\normalfont\normalsize\bfseries}{}{0pt}{#1}
\titlespacing*{\section}{0pt}{2.2ex plus 0.6ex minus 0.3ex}{1.2ex}
\titlespacing*{\subsection}{0pt}{1.8ex plus 0.5ex minus 0.3ex}{0.9ex}
\titlespacing*{\paragraph}{0pt}{1.2ex plus 0.4ex minus 0.2ex}{0.8em}

\makeatletter
\renewcommand{\maketitle}{%
  \thispagestyle{plain}%
  \begin{center}
  {\LARGE\bfseries Detection Is Not Enough:\\[3pt] Reasoning Traces as Testable Program Hypotheses\par}
  \vspace{1.4em}
  {\large Vaidehi Hegde}\\[3pt]
  Indian Institute of Technology (Indian School of Mines) Dhanbad\\
  \texttt{vaidehihegde05@gmail.com}
  \end{center}
  \vspace{1.2em}}
\renewenvironment{abstract}{%
  \begin{center}{\large\bfseries Abstract}\end{center}\vspace{-0.4em}%
  \begin{list}{}{\setlength{\leftmargin}{0.5in}\setlength{\rightmargin}{0.5in}\setlength{\listparindent}{0pt}\setlength{\parsep}{0pt}}\item\relax}
  {\end{list}\vspace{0.6em}}
\makeatother

\hypersetup{pdftitle={Detection Is Not Enough: Reasoning Traces as Testable Program Hypotheses},pdfauthor={Vaidehi Hegde}}

\begin{document}
\maketitle

\begin{abstract}
When a language model is deployed, no reference solution exists, so the reasoning trace it writes is the only evidence for judging its answer.
The reference-free substitutes each fail differently: self-assessment degrades the answer it was meant to protect, sampling agreement measures stability rather than correctness, and perturbation studies read the expected answer off a human-written template, grading a model rather than a trace.
We ask how far a single trace can be verified with no reference solution, no judge model and no weight update, and find, for a frozen LLaMA-2-7B-chat on GSM8K, that detection is not what limits such a pipeline.
Checking a trace against itself reaches an AUROC of 0.741 and stops there for a structural reason: two thirds of wrong traces state no false equation, having executed a mistaken plan correctly, and on wrong-plan failures the detector flags a wrong trace less often than a right one.
A candidate-derived metamorphic probe exposes many of these failures: we abstract the trace into a program over the question's quantities, perturb one quantity, execute the program, and ask the model the rewritten question once, so every trace supplies its own test.
Valid perturbations exist for 937 of 1,319 test questions (71.0\%); among the 847 that also have a parseable first trace, supported traces are correct 60.6\% of the time against a 31.6\% base rate, although the probe also refutes 45.5\% of correct traces.
No tested policy yields a statistically significant improvement in accuracy.
A flag earns its retry only when detector precision exceeds $D/(R+D)$, where $R$ and $D$ are the rates at which a retry repairs a wrong trace and damages a correct one; our repair rate is so low that, at the same retry budget, even a perfect detector would recover an estimated 66 of 1,319 questions.
Repair, not detection, is the open problem.
\end{abstract}

\section{Introduction}
\label{sec:intro}

Large language models now solve multi-step mathematical problems by writing out their reasoning before they answer \citep{wei2022cot}.
Whether that reasoning can be trusted is usually decided with labelled solutions, used either to train a verifier on final answers \citep{cobbe2021gsm8k} or to supervise individual steps \citep{uesato2022process,lightman2023verify}.
At deployment those labels are absent, and the trace is often the only evidence available for judging the answer it accompanies.

\begin{figure}[t]
\centering
\includegraphics{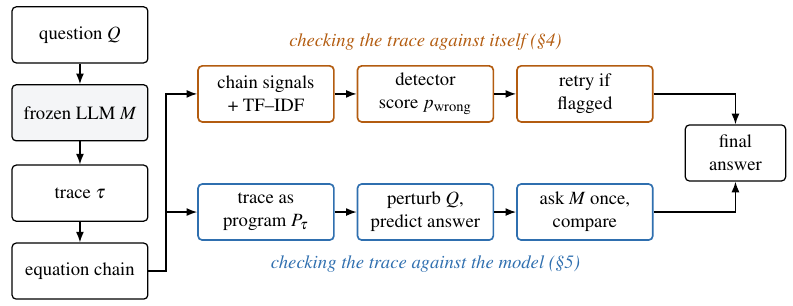}
\caption{Overview of the pipeline. A frozen model answers a question with a trace, which is parsed into a chain of arithmetic claims. The static detector (orange) decides from the trace alone whether to spend a retry. The metamorphic probe (blue) reads the trace as a program, predicts the answer to a perturbed question, and compares that prediction with the model's own answer. No component uses a reference answer or a judge model.}
\label{fig:overview}
\end{figure}

Label-free substitutes exist, but each has a known weakness.
Asking a model to critique and revise its own answer without external feedback lowers GSM8K accuracy below the first attempt \citep{huang2024selfcorrect,kamoi2024survey}, and small models need a strong external verifier before self-correction helps at all \citep{zhang2024small}.
Agreement across samples \citep{wang2023selfconsistency} measures the stability of an answer rather than its correctness, and a model that misreads a question tends to misread it the same way every time.
Internal-consistency verification instead judges a trace by whether its own arithmetic and dependencies hold together \citep{zhong2026verytrace}, and rests on the assumption that a wrong trace will, somewhere, look wrong.
We ask how far a single trace can be verified with no reference solution, no judge model and no weight update, and what the answer implies for a system that uses a verifier to decide when to try again.

We study this question on a frozen LLaMA-2-7B-chat \citep{touvron2023llama2} on GSM8K \citep{cobbe2021gsm8k}, instrumented so that every measurement is made by regular expressions and a computer algebra system (\cref{fig:overview}).
We first check traces against themselves: a static detector built from the equations a trace states reaches an AUROC of 0.741 on test, and its limit lies in the traces, since 67.1\% of failures contain no false equation.
We then check traces against the model: we read the trace as a program over the question's quantities, change one quantity, execute the program to predict the new answer, and ask the model the changed question once.
Unlike perturbation benchmarks \citep{mirzadeh2024gsmsymbolic,zhou2025checklist,zhou2026lgmt}, which take the expected answer from a hand-written template and grade a model, this probe derives its expected answer from the candidate and grades a single trace. A valid perturbation exists for 71.0\% of test questions, and our probe statistics are computed on the 847 of these that also have a parseable first trace.
Finally, we ask why better detection does not improve accuracy, and show that a flag is worth a retry only when detector precision exceeds $D/(R+D)$, where $R$ and $D$ are the rates at which a retry repairs a wrong answer and damages a correct one.

\Cref{tab:summary} summarises the findings.
The probe exposes many of the failures that self-checking misses, but for this model the binding constraint is repair: most failures contain no correct prefix to preserve, so a retry is little more than a fresh attempt at a question the model has just failed.
We therefore argue that reference-free verifiers should be evaluated by the repair and damage rates of the action they trigger, not by detection metrics alone.
All measurements come from one model, one dataset and one prompting setup, and the model is weak by current standards; \cref{sec:limitations} separates the conclusions we expect to transfer from those we do not.

\begin{table}[t]
\centering
\caption{Summary of findings.}
\label{tab:summary}
\small
\begin{tabular}{@{}p{0.27\textwidth}p{0.56\textwidth}l@{}}
\toprule
question & finding & section \\
\midrule
How far does checking a trace against itself go? & It stalls at AUROC 0.741 because 67.1\% of failures contain no false equation; on wrong-plan failures the detector fires less often than on correct traces. & \S\ref{sec:static} \\
\addlinespace
Can the missing evidence be obtained without a judge? & Partly. 937 of 1,319 questions (71.0\%) admit a valid probe. On the 847 of these with a parseable first trace, the probe refutes 81.4\% of wrong-plan and misread-question failures, and supported traces are correct 60.6\% of the time against 31.6\% in that subset. & \S\ref{sec:probe} \\
\addlinespace
Does better detection improve accuracy? & Not in our tests: no policy improves accuracy significantly. A flag pays only above precision $D/(R+D)$, and at the measured repair rate of 12.3\% a perfect detector on the same budget would gain an estimated five points. & \S\ref{sec:accuracy} \\
\bottomrule
\end{tabular}
\end{table}

\section{Background and Related Work}
\label{sec:background}

Verifying a reasoning trace means deciding, from the trace and possibly further computation, whether the answer it reaches is correct.
The literature uses the term for settings that differ in what the verifier may consult, which makes results difficult to compare.
We therefore fix the setting first and then place prior work within it.

\begin{definition}[Reference-free verifier]
A reference-free verifier is a function of the question $Q$, a trace $\tau$ produced by a frozen model $M$, and possibly further calls to the same $M$, that uses neither the reference answer, a separately trained judge language model, nor any update to the weights of $M$.
Throughout this paper, \emph{reference-free} refers to inference-time verification: the deployed procedure does not consult a reference answer for the test question. The static detector is trained using labelled training traces, whereas the metamorphic probe derives its verdict without reference-answer supervision.
\end{definition}

Under this definition, a verifier is useful only through what it causes.
We study the most common use, a \emph{detect-then-retry} policy, which flags a trace when the verifier's score passes a threshold, generates a second trace for each flagged question, and either returns the second answer (\emph{replace}) or lets a reference-free selector choose between the two (\emph{select}).

\Cref{tab:notation} collects the notation used in the rest of the paper.

\begin{table}[t]
\centering
\caption{Notation.}
\label{tab:notation}
\small
\begin{tabular}{@{}ll@{}}
\toprule
symbol & meaning \\
\midrule
$Q$, \ $\bm q=(q_0,\dots,q_k)$ & question and its numeric quantities \\
$M$, \ $\tau=M(Q)$, \ $\hat y(\tau)$ & frozen model, its trace, the trace's final answer \\
$\mathcal E(\tau)=(e_1,\dots,e_m)$ & equation chain; each claim $e_j$ is $\ell_j=c_j$ \\
$\alpha$, \ $P_\tau$ & abstraction map and the program it extracts (\cref{def:abstraction}) \\
$Q'$, \ $\bm q'$ & perturbed question and its quantities \\
$\pwrong(\tau)$ & static detector's estimate that $\tau$ is wrong \\
$\Nflag$, \ $\pi$ & number of flagged questions, and the fraction of them that are wrong \\
$R$, \ $D$ & repair rate on flagged wrong traces, damage rate on flagged correct traces \\
$\pistar=D/(R+D)$ & break-even precision (\cref{prop:identity}) \\
\bottomrule
\end{tabular}
\end{table}

Learned verifiers rank solutions by a model trained on final-answer labels \citep{cobbe2021gsm8k} or step-level labels \citep{uesato2022process,lightman2023verify}; they require annotation and a verifier model, and fall outside our setting.
Self-correction asks the model to critique and refine its own output \citep{madaan2023selfrefine}, and has been shown not to improve reasoning when nothing external indicates which answers need revising \citep{huang2024selfcorrect,kamoi2024survey}; the model acts as its own judge.
Self-consistency \citep{wang2023selfconsistency} estimates reliability through repeated answers to the same question.
Other reference-free approaches examine consistency between a model's direct answer and the answer obtained by composing a stated decomposition \citep{bottman2026operadic}, or combine step validity, support of the answer by the reasoning, and counterfactual sensitivity into a broader diagnostic \citep{shukla2026silent}.
These approaches differ in the evidence they use and in the aspects of reasoning they test.
Internal-consistency verification judges the trace directly; VeryTrace \citep{zhong2026verytrace} compiles a trace into a structured formalism, checks its calculations, dependencies and conclusion, and delegates the rest to a second model.
Our first question is what remains when that second model is removed.

Metamorphic testing validates software without an oracle by checking relations between outputs on related inputs \citep{chen1998metamorphic,segura2016survey,chen2018mtreview}, and behavioural testing brought predefined perturbations to NLP \citep{ribeiro2020checklist}.
For mathematical reasoning, GSM-Symbolic \citep{mirzadeh2024gsmsymbolic} regenerates GSM8K from templates and finds that accuracy drops when values change, MathCheck \citep{zhou2025checklist} and related diagnostics \citep{sahoo2025catch} report similar fragility, and LGMT \citep{zhou2026lgmt} counts self-contradictions across semantically equivalent rewrites.
Each of these characterises a model: the relation is written by a person and the expected answer comes from the template.
Self-verification \citep{weng2023selfverification} derives its test from the candidate, as we do, but asks a further generation for the verdict.
Finally, a trace need not describe the computation that produced its answer \citep{turpin2023unfaithful,lanham2023faithfulness,lyu2023faithful}; the probe we introduce can be read as a behavioural test of a weak form of this faithfulness.
\Cref{tab:compare} places these approaches against our definition.

\begin{table}[t]
\centering
\caption{Verification approaches for a single deployed answer. \cmark{}~yes, \xmark{}~no, \pmark{}~partly. ``No reference'' means the expected answer comes from neither a label nor a template; ``one trace'' means the verdict concerns the specific trace rather than the model or an answer distribution. The table compares properties, not performance.}
\label{tab:compare}
\small
\setlength{\tabcolsep}{5pt}
\begin{tabular}{@{}lccccc@{}}
\toprule
approach & \makecell{no\\reference} & \makecell{no trained\\verifier} & \makecell{no judge\\LM} & \makecell{one\\trace} & \makecell{extra\\generations}\\
\midrule
Outcome / process verifiers \citep{cobbe2021gsm8k,lightman2023verify} & \cmark & \xmark & \xmark & \cmark & 0 \\
Self-correction \citep{huang2024selfcorrect,madaan2023selfrefine} & \cmark & \cmark & \xmark & \cmark & $\ge 1$ \\
Self-consistency \citep{wang2023selfconsistency} & \cmark & \cmark & \cmark & \xmark & $k$ \\
Template perturbation \citep{mirzadeh2024gsmsymbolic,zhou2026lgmt} & \xmark & \cmark & \cmark & \xmark & many \\
Self-verification \citep{weng2023selfverification} & \cmark & \cmark & \xmark & \cmark & $\ge 1$ \\
VeryTrace \citep{zhong2026verytrace} & \cmark & \cmark & \pmark & \cmark & $\ge 1$ \\
\midrule
Static detector (\S\ref{sec:static}) & \cmark & \pmark & \cmark & \cmark & 0 \\
Metamorphic probe (\S\ref{sec:probe}) & \cmark & \cmark & \cmark & \cmark & $\approx 1$ \\
\bottomrule
\end{tabular}
\end{table}

\paragraph{What is new here.}
The closest prior work either verifies a trace from the inside \citep{zhong2026verytrace}, derives a check from the candidate but asks another generation for the verdict \citep{weng2023selfverification}, perturbs questions to evaluate a model against a template answer \citep{mirzadeh2024gsmsymbolic,zhou2026lgmt}, compares a direct answer with a composed one \citep{bottman2026operadic}, or includes counterfactual sensitivity as one component of an aggregate trace score \citep{shukla2026silent}.
Our mechanism combines elements of these in a specific way: we \emph{extract a program from the trace}, \emph{execute it} on a perturbed version of the question to obtain the answer the trace itself predicts, and test that prediction against the \emph{same model's} answer to the perturbed question.
The expected answer therefore comes neither from a label nor from a template, and the verdict comes neither from a judge nor from agreement among samples.
Operadic consistency \citep{bottman2026operadic} also uses the model's own decomposition, but it asks whether two routes to the \emph{same} question agree; our test changes the question and asks whether the trace's procedure, executed symbolically, predicts what the model does on the new one.
RAFS \citep{shukla2026silent} treats counterfactual sensitivity as one signal among several; in our probe the counterfactual is an executable numerical prediction with a single expected value, obtained by running the program extracted from the candidate, and the probe is studied on its own so that its effect on each failure class can be isolated.
This also distinguishes the probe from self-consistency: self-consistency compares answers to the same question and rewards an answer that is stable whether or not it is right, whereas the probe compares a prediction about a \emph{different} question with the model's behaviour on it, and so tests whether the stated procedure is the one the model follows.
We do not claim that perturbation testing or program extraction is new in itself.
Our contribution is this candidate-derived counterfactual test, together with two analyses that go beyond the mechanism: which failure classes it exposes that internal consistency cannot, and why improved detection need not improve final accuracy.

In the following sections we evaluate both kinds of reference-free verification and the pipeline they drive.
\Cref{sec:static} shows that checking a trace against itself cannot detect most failures.
\Cref{sec:probe} shows that testing the trace against the model's behaviour reaches those failures.
\Cref{sec:accuracy} shows that neither yields a significant improvement in accuracy, and explains why.

\section{Experimental Setup}
\label{sec:setup}

\paragraph{Model and data.}
All experiments use a frozen LLaMA-2-7B-chat \citep{touvron2023llama2} on GSM8K \citep{cobbe2021gsm8k}.
We chose this model because it has open weights, can be run frozen and fully instrumented on a single GPU, and fails often enough (on 72\% of test questions at the first attempt) to provide large samples of every failure type; the cost of this choice is representativeness, which we discuss in \cref{sec:limitations}.
We report all headline numbers on the 1,319 test questions.
The training split supplies labelled traces for fitting the static detector (2,568 traces, comprising first attempts and retries), the 1,068 first-attempt failures used for the failure taxonomy and for measuring repair rates, and a development set of 400 further questions held out from fitting.
Reference answers label traces on the training split only and never enter a context the model sees.

\paragraph{Prompting and retries.}
Both attempts draw few-shot exemplars from a pool of the model's own correct traces, which grows during the training run and is frozen before any test question is seen.
The first attempt, which we call \emph{pass-1}, retrieves exemplars by the question alone and solves 27.8\% (367 of 1,319) of test questions.
A retry additionally conditions its retrieval on a symbolic diagnosis of how the first attempt failed (\cref{app:taxonomy}), so it is not a plain resample.

\paragraph{Measurement.}
We parse each trace into its \emph{equation chain} $\mathcal E(\tau)=(e_1,\dots,e_m)$, the ordered list of arithmetic claims $e_j:\ell_j=c_j$ it states, and sympy marks each claim true or false.
No language model acts as a grader anywhere in scoring, diagnosis or verification.

\paragraph{Evaluation.}
Classifiers are evaluated with GroupKFold using the question as the group, so that a question's first attempt and its retry never fall on opposite sides of a split, and an assertion rejects any feature derived from the reference solution; both safeguards were added after we found the problems they prevent (\cref{app:audit}).
Paired accuracy comparisons use exact McNemar tests on the same 1,319 questions, and contingency comparisons use Fisher's exact test.
Every accuracy we report is what a deployed system would return: we never count a question as solved because one of two candidates was right, since choosing between them would require the reference answer.
\Cref{fig:overview} shows the full pipeline.

\section{Internal Consistency Cannot Detect Most Failures}
\label{sec:static}

In this section we build a static detector that checks a trace against itself, measure how well it works, and examine why it stops where it does.

\subsection{The static detector}

From the equation chain and the question we compute four signals, none of which reads a reference answer (\cref{app:signals}).
\emph{Equation validity} is the fraction of stated equations that sympy confirms; \emph{value overlap} compares the numbers the question gives with those the trace uses; \emph{goal engagement} records whether the trace ever computes the quantity the question asks for; and \emph{coherence} checks that every step refers only to values that were given or previously derived.
We concatenate these signals with TF--IDF features of the trace text and blend a logistic regression on the sparse features with a gradient-boosted model on the dense ones.
The training target is whether the final answer is wrong, and a second head predicts the type of failure.

\subsection{Results}

\begin{table}[t]
\centering
\caption{Results of the static detector. Out-of-fold estimates use GroupKFold by question on the training split.}
\label{tab:detector}
\small
\begin{tabular}{@{}lccc@{}}
\toprule
 & out of fold & development (400) & test (1,319) \\
\midrule
detection AUROC & 0.781 & 0.770 & 0.741 \\
type-head accuracy & 0.598 & -- & 0.590 \\
\bottomrule
\end{tabular}
\end{table}

\Cref{tab:detector} reports the detector's performance.
It reaches an AUROC of 0.781 out of fold, 0.770 on the development set and 0.741 on test, and its type head reaches 0.590 on test against 0.598 estimated.
At its operating cut of $\pwrong\ge0.683$ it flags 540 test questions (40.9\%) with 87.6\% precision.
Detection AUROC is 0.04 lower on test than out of fold, while the type head is within one point of its out-of-fold estimate; \cref{sec:why-static} examines why detection stops at this level.

\subsection{Why does detection stall?}
\label{sec:why-static}

\paragraph{Empirical analysis.}
Only 32.9\% of pass-1 failures contain a locatable false equation; in the remaining 67.1\%, every equation the trace states is true.
\Cref{fig:taxonomy} shows where these failures come from.
The two largest classes, a wrong plan executed correctly (\type{WP}) and a misread question solved correctly (\type{SM}), account for 47.5\% of training failures, and both have sound arithmetic.
On test, the detector flags 12\% of \type{WP} failures, which is below the 18.3\% of correct traces it flags (\cref{fig:per-type}); on this class its signal points in the wrong direction.

\begin{figure}[t]
\centering
\includegraphics{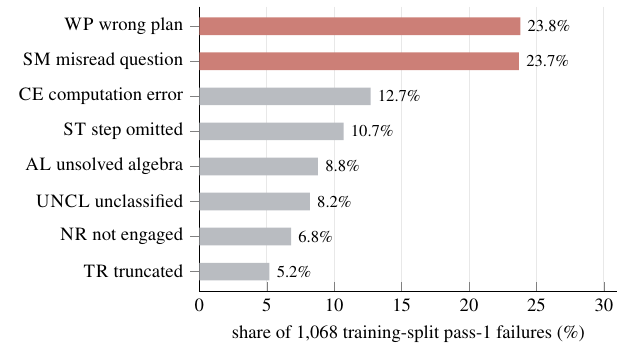}
\caption{Distribution of failure types over the 1,068 training-split pass-1 failures. The two largest classes (red) have sound arithmetic and leave no symbolic sign of the error. Definitions and per-class recall of the type head are in \cref{tab:taxonomy}.}
\label{fig:taxonomy}
\end{figure}

\paragraph{An illustrative example.}
\Cref{fig:example} shows a test question the detector cleared.
The trace multiplies the 16 eggs by the 3 eaten at breakfast, a step the problem never calls for, and then proceeds without error.
All four equations hold, every quantity in the question is used, each step consumes a value that was given or derived earlier, and the trace ends by computing an amount of money as asked.
The detector has nothing to object to and scores the trace $\pwrong=0.581$, below its cut, so no retry is spent.

\paragraph{Intuitive explanation.}
A check of a trace against itself can detect an error only if the error leaves a mark inside the trace, such as a false equation, an unused quantity or an unexplained number.
A wrong plan executed correctly leaves no such mark, because its error lies in which arithmetic was chosen rather than in how it was carried out.
For the structural signals and classifiers we evaluated, a richer model over the same evidence would therefore have little to work with on these failures.
This is a limitation of the evidence we extract, not an impossibility result for all reference-free verification, but it suggests that the information separating such a trace from a correct one is most readily found outside the trace.

\begin{figure}[t]
\centering
\fcolorbox{midgray!50}{softgray}{\parbox{0.95\linewidth}{\small
\textbf{Question.} Janet's ducks lay 16 eggs per day. She eats three for breakfast and bakes muffins with four. She sells the remainder at \$2 per egg. How much does she make? \hfill \textcolor{midgray}{(reference: 18)}\\[4pt]
\textbf{Pass-1 trace.} \ $16\times 3=48$; \ $48-3=45$; \ $45-4=41$; \ $2\times 41=82$. \ The answer is 82. \hfill \textcolor{badred}{\textbf{incorrect}}}}

\vspace{8pt}
\includegraphics{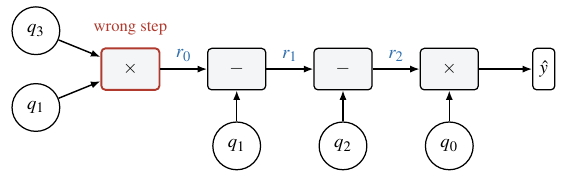}

\vspace{6pt}
\small
\begin{tabular}{@{}lcccc@{}}
\toprule
question asked & program predicts & model answers & true answer & verdict \\
\midrule
original ($q_0{=}2$, $q_3{=}16$) & 82 & 82 & 18 & -- \\
price \$2 $\to$ \$5 \ ($q_0$) & 205 & 450 & 45 & \textcolor{badred}{refute} \\
eggs 16 $\to$ 32 \ ($q_3$) & 178 & 72 & 50 & \textcolor{badred}{refute} \\
\bottomrule
\end{tabular}
\caption{A wrong-plan failure from the GSM8K test set. Every equation in the trace is true, so the static detector clears it ($\pwrong=0.581$, below the cut of 0.683). The abstraction of \cref{sec:abstraction} turns the trace into the program shown, with question slots $q_0{=}2$, $q_1{=}3$, $q_2{=}4$, $q_3{=}16$ in the parser's ordering; the first step, multiplying the eggs by the number eaten, is the wrong plan. Under both perturbations the answer the program predicts differs from the model's answer, so both probes refute the trace. The true answers are shown for illustration only; the probe never uses them.}
\label{fig:example}
\end{figure}

\section{Traces Can Be Tested as Program Hypotheses}
\label{sec:probe}

A failure with sound arithmetic contains no evidence of its own error, but the trace that states it also makes a claim about what the model did.
In this section we read that claim as a program and test it from outside the trace, which is metamorphic testing \citep{chen1998metamorphic} with the relation supplied by the candidate.

\subsection{From trace to program}
\label{sec:abstraction}

\begin{definition}[Abstraction]
\label{def:abstraction}
Given the chain $(e_1,\dots,e_m)$ with $e_j:\ell_j=c_j$, the abstraction $\alpha$ rewrites every numeric literal $x$ in $\ell_j$ as
\[
\alpha(x)=
\begin{cases}
q_i & \text{if } x \text{ equals question quantity } q_i,\\
r_{j'} & \text{if } x=c_{j'} \text{ for an earlier step } j'<j,\\
x & \text{otherwise},
\end{cases}
\]
and binds $r_j:=\alpha(\ell_j)$. The program $P_\tau$ is the resulting straight-line code, whose output is the variable bound to the trace's final answer.
\end{definition}

For the trace in \cref{fig:example}, the abstraction yields $r_0=q_3q_1$, $r_1=r_0-q_1$, $r_2=r_1-q_2$ and $\hat y=q_0r_2$.
We regard $P_\tau$ as a hypothesis about the procedure the trace describes, not as a claim about the model's internal computation.

\subsection{The metamorphic probe}

\begin{algorithm}[t]
\caption{Metamorphic probe of a single trace. No step reads a reference answer.}
\label{alg:probe}
\begin{algorithmic}[1]
\Require question $Q$ with quantities $\bm q$, trace $\tau$, frozen model $M$
\State $\mathcal E \gets \textsc{ParseChain}(\tau)$; \ $P_\tau \gets \alpha(\mathcal E)$
\State choose $q_i$ and a replacement $q_i'$ that satisfy all guards; \textbf{if} none exists \textbf{return} \textsc{abstain}
\State $Q' \gets Q$ with its single occurrence of $q_i$ rewritten to $q_i'$
\State $\hat y' \gets P_\tau(\bm q')$ \Comment{symbolic execution; changes propagate downstream}
\State $y' \gets \hat y(M(Q'))$ \Comment{one generation}
\State \Return \textsc{support} if $y'=\hat y'$, otherwise \textsc{refute}
\end{algorithmic}
\end{algorithm}

\Cref{alg:probe} gives the procedure.
The probe rewrites one quantity in the question, executes the program under that substitution to obtain the answer the trace predicts, and asks the frozen model the rewritten question once.
Agreement supports the hypothesis that the trace describes a procedure consistent with the model's behaviour on the perturbed question. Disagreement indicates a mismatch, but does not by itself establish that the trace is unfaithful, since the model may answer the perturbed question incorrectly.
For the example in \cref{fig:example}, raising the price from \$2 to \$5 makes the program predict 205 while the model answers 450, and raising the eggs from 16 to 32 makes it predict 178 while the model answers 72, so both probes refute the trace.
Because the expected answer comes from the candidate's own chain, the probe needs no reference answer at test time, which distinguishes it from template-based perturbation, where the expected answer is the template's.

\paragraph{Guards.}
A probe is worth a generation only if it can discriminate.
The most important guard requires the perturbed quantity to be \emph{load-bearing}, meaning that some candidate's chain consumes it; changing an unused number should leave the answer unchanged whatever the trace says, so good and bad traces would appear equally consistent.
Four further guards (\cref{app:guards}) require the quantity to occur exactly once in the question, the replacement not to collide with another number, integrality and sign to be preserved, and at least one candidate's prediction to move.
Under these guards, 937 of the 1,319 test questions (71.0\%) admit a probe and the remainder abstain (\cref{fig:coverage}).
Of the 937, 847 also have a parseable pass-1 answer, and all detection results in this section are computed on these 847.
The run used 988 probes in total: 886 of the 937 testable questions received one probe and 51 received two.
Traces that never engage with the question (\type{NR}) consume none of its quantities, and, as the load-bearing guard implies, none of them is among the testable traces.

\begin{figure}[t]
\centering
\includegraphics{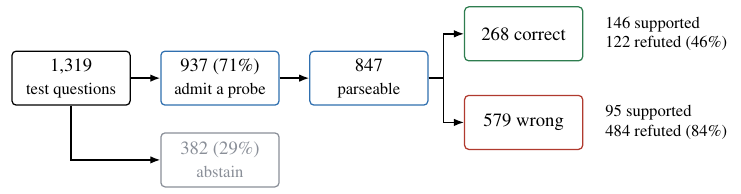}
\caption{Coverage of the metamorphic probe on the test set, and its verdicts. Detection results use the 847 testable questions with a parseable pass-1 answer; accuracy results use all 1,319 questions, with abstentions falling back to the non-probe policy.}
\label{fig:coverage}
\end{figure}

\subsection{Results}

\paragraph{Does a faithful trace predict its own perturbation?}\label{sec:faithful}
The probe is meaningful only if a trace that describes the right procedure also predicts the right answer to the perturbed question.
We check this offline on 1,500 training traces, using the inline calculator annotations in GSM8K to compute the true perturbed answer; this check is a diagnostic and plays no part at inference.
Correct traces predict the true perturbed answer 77.9\% of the time and wrong traces 7.2\%.
For this model, then, a trace that reaches the right answer usually also states how it got there, which chain-of-thought does not guarantee in general \citep{turpin2023unfaithful,lanham2023faithfulness}.

\begin{table}[t]
\centering
\caption{Results of the metamorphic probe on the 847 testable pass-1 traces. Left: verdicts by correctness. Right: the verdicts read as diagnostic tests.}
\label{tab:support}
\small
\begin{tabular}{@{}lrrrr@{\hspace{2.2em}}lr@{}}
\toprule
verdict & correct & wrong & total & $P(\text{correct})$ & diagnostic & value\\
\midrule
supported & 146 & 95 & 241 & 60.6\% & refutation rate, wrong traces & 83.6\% \\
refuted & 122 & 484 & 606 & 20.1\% & refutation rate, correct traces & 45.5\% \\
\cmidrule(r{2.2em}){1-5}
all & 268 & 579 & 847 & 31.6\% & likelihood ratio, support & 3.32 \\
 & & & & & likelihood ratio, refutation & 1.84 \\
 & & & & & odds ratio (Fisher $p=5.6\times10^{-29}$) & 6.1 \\
\bottomrule
\end{tabular}
\end{table}

\paragraph{Support and refutation.}
\Cref{tab:support} reports the verdicts.
A single supporting probe nearly doubles the probability that a trace is correct, from a base rate of 31.6\% to 60.6\% (odds ratio 6.1, Fisher $p=5.6\times10^{-29}$).
Refutation catches 83.6\% of wrong traces but also refutes 45.5\% of correct ones, so on its own it is weak evidence, with a likelihood ratio of 1.84 against 3.32 for support.
In an exploratory analysis, combining the probe score with the static score at equal weights raises detection AUROC on these 847 questions from 0.740 to 0.773, and across blend weights from 0.1 to 0.9 the blended AUROC ranges from 0.769 to 0.778, so the improvement does not depend on the choice of weight. Because the weight grid was first examined on the test set (\cref{app:audit}), this is not a clean held-out estimate.

\begin{figure}[t]
\centering
\includegraphics{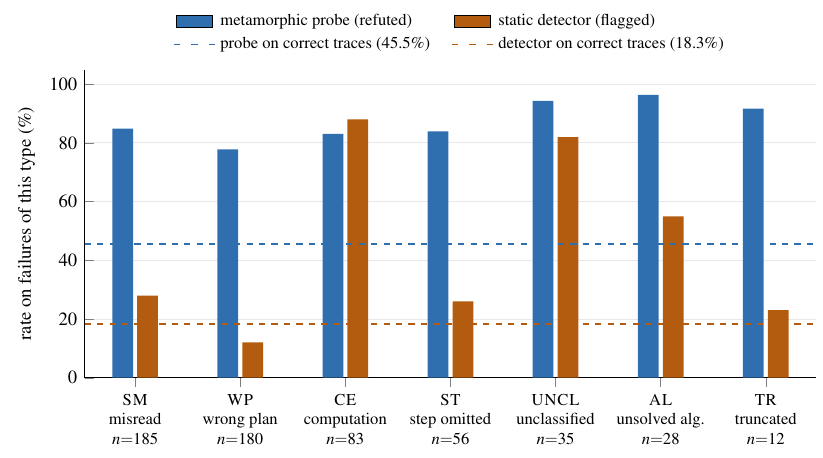}
\caption{Rate at which each verifier calls a trace wrong, by failure type, on the 579 wrong testable traces. Dashed lines give each verifier's rate on the 268 correct traces; a verifier is informative on a class only where its bar rises above its own line. The probe does so on every class (all per-class Fisher $p<0.002$, \cref{app:arms}); the static detector falls below its line on \type{WP}.}
\label{fig:per-type}
\end{figure}

\paragraph{Where the two verifiers differ.}
\Cref{fig:per-type} breaks the verdicts down by failure type.
Because each verifier pays its own false-positive rate on correct traces, what matters is how far each rises above its own baseline on each class.
On the structural classes, such as computation errors and unsolved algebra, the static detector is strong, because these failures leave a symbolic mark.
On the semantic classes it is weak or inverted, flagging 28\% of \type{SM} and 12\% of \type{WP} failures against 18.3\% of correct traces.
The probe refutes 157 of 185 \type{SM} failures and 140 of 180 \type{WP} failures, 81.4\% of the two classes together against 45.5\% of correct traces ($p=4.2\times10^{-21}$).
In short, the static detector sees errors in the arithmetic, while the probe sees errors in the choice of arithmetic.

\subsection{Why is refutation weak evidence?}
\label{sec:asymmetry}

The model answers perturbed questions no more reliably than original ones (28.6\% against 27.8\%).
Support therefore requires two things at once, a faithful program and a matching answer to $Q'$, whereas refutation can arise from either an unfaithful program or an unreliable answer.
The example in \cref{fig:example} shows the second case: neither of the model's perturbed answers is correct (the true values are 45 and 50), so the refutations there reflect the model's unreliability as much as the trace's.
Treating the two conditions as independent gives a per-probe support rate of roughly $0.779\times0.286\approx0.22$ for a correct trace and $0.072\times0.286\approx0.02$ for a wrong one.
The observed rates are higher in both groups (54.5\% and 16.4\%) because correctness on $Q'$ is correlated with correctness on $Q$, but the shape is the same: support is far more likely for a correct trace, while refutation is common for both.
We therefore allow support to promote a candidate and never allow refutation to demote one.
For the same reason we do not use refutation to trigger retries: with a likelihood ratio of 1.84 it would flag many correct traces, and at a damage rate of 74.6\% each such flag is costly (\cref{sec:accuracy}).
For the same reason, a marginal AUROC understates the probe, whose value lies in the precision it adds to the traces it supports.
\Cref{sec:accuracy} shows that precision is the quantity on which the usefulness of a verifier turns.

\subsection{What can the probe not detect?}
\label{sec:blind}

The probe tests whether a trace's procedure is consistent with the model's behaviour on a perturbed question, not whether that procedure is right.
This distinction defines its blind spot.
If the model applies the same wrong procedure to the perturbed question that it applied to the original, the program extracted from the trace predicts exactly what the model answers, and the probe returns support for a wrong trace.
A misreading that the model repeats under every rewording, which defeats self-consistency, can therefore defeat the probe as well.
The probe is useful on this model because wrong traces are rarely stable in this sense: only 7.2\% of wrong traces predict the true perturbed answer, against 77.9\% of correct ones (\cref{sec:faithful}), and only 95 of the 579 wrong testable traces (16.4\%) were supported.
Those 95 traces are the probe's misses, and they are concentrated where we would expect: among \type{WP} failures, 40 of 180 were not refuted, the highest miss rate of any class (\cref{fig:per-type}).
A model whose errors are more systematic than this one's would give the probe less to work with, and the separation in \cref{tab:support} should not be assumed to carry over without measurement.
Two further blind spots are structural.
Questions in which no load-bearing quantity can be perturbed safely receive no probe at all (29\% of the test set), and a trace whose chain cannot be parsed yields no program.

\section{Better Detection Does Not Necessarily Improve Accuracy}
\label{sec:accuracy}

In this section we measure the accuracy of every policy built on the two verifiers, derive the condition under which a verifier can improve accuracy at all, and explain why that condition is hard to meet for this model.

\subsection{Results}

\begin{table}[t]
\centering
\caption{End-to-end accuracy of every policy on the 1,319 test questions, with exact McNemar tests against pass-1.}
\label{tab:arms}
\small
\begin{tabular}{@{}lrrr@{}}
\toprule
policy & correct & accuracy & $p$ \\
\midrule
pass-1 & 367 & 27.82\% & -- \\
replace flagged & 375 & 28.43\% & 0.50 \\
replace flagged, with metamorphic veto & 374 & 28.35\% & 0.55 \\
select, symbolic & 379 & 28.73\% & 0.21 \\
select, metamorphic & 382 & 28.96\% & 0.091 \\
select, metamorphic + back-substitution & 382 & 28.96\% & 0.096 \\
\bottomrule
\end{tabular}
\end{table}

\Cref{tab:arms} reports the results.
Replacing every flagged answer gains 8 questions, and letting metamorphic evidence choose between the two candidates gains 15, reaching 28.96\% against 27.82\%.
No policy differs significantly from pass-1; the best has $p=0.091$.
Probe evidence decided only 90 of the 1,319 questions, because support is rare where it would matter.
On the development set, retrying the questions the detector flagged gained 4.8 points while retrying a random set of the same size lost 2.7; this comparison was not repeated on test, and the development set also informed calibration, so we treat it as supporting evidence only.

\subsection{When is a flag worth a retry?}

A second attempt can repair a wrong answer, but it can also break a correct one.
Suppose a detector flags $\Nflag$ questions with precision $\pi$, and let $R$ be the fraction of flagged wrong traces whose retry is correct and $D$ the fraction of flagged correct traces whose retry is wrong.

\begin{proposition}
\label{prop:identity}
If every flagged answer is replaced by its retry, the number of correctly answered questions changes by
\begin{equation}
\label{eq:delta}
\Delta=\Nflag\bigl(\pi R-(1-\pi)D\bigr),
\end{equation}
which is positive if and only if
\begin{equation}
\label{eq:breakeven}
\pi>\pistar=\frac{D}{R+D}.
\end{equation}
\end{proposition}

\begin{proof}
Unflagged questions are unaffected. Flagged wrong traces become correct at rate $R$, contributing $\pi\Nflag R$, and flagged correct traces become wrong at rate $D$, contributing $-(1-\pi)\Nflag D$. Rearranging $\pi R>(1-\pi)D$ gives \eqref{eq:breakeven}.
\end{proof}

When $R$ and $D$ are measured on the flagged set itself, \eqref{eq:delta} is an exact count rather than a model.
On test, the detector flags 540 questions, of which 473 are wrong and 67 correct; retries repair 58 of the former ($R=12.3\%$) and damage 50 of the latter ($D=74.6\%$), giving $\Delta=58-50=+8$, the observed change.
The break-even precision is therefore $\pistar=85.9\%$, which the detector's 87.6\% exceeds by 1.7 points.
On the training split, the calibration (\cref{tab:calibration}) uses repair and damage rates of $R=16.1\%$ and $D=76.5\%$ from the training run, a break-even of 82.6\%.
The difference lies in repair: on the training split the diagnosis that conditions a retry comes from a symbolic cascade and is essentially always right, whereas on test it comes from the type head, which is right 59\% of the time.

Using \eqref{eq:delta} counterfactually, to ask what a different detector would achieve, requires holding $R$ and $D$ fixed, which assumes that repairability does not depend on which wrong traces are flagged.
This is a real assumption, since a better detector would newly flag the kinds of failures the current one misses, and the estimates below should be read as first-order.
With that caveat, and holding the budget at 540 retries and the repair rate at its measured value, a perfect detector would gain an estimated $540\times0.123\approx66$ questions (+5.0 points), and one allowed to flag all 952 failures would gain about 117 (+8.9 points).
The two partial derivatives of \eqref{eq:delta} are similar, $\partial\Delta/\partial\pi=\Nflag(R+D)\approx469$ and $\partial\Delta/\partial R=\Nflag\pi\approx473$, so detection and repair are equally effective at the margin.
What differs is the room each has.
Precision can rise by at most 12.4 points, worth at most 58 questions, whereas the repair rate can rise by up to 87.7 points, worth about 415 questions at the current precision.
Doubling $R$ alone would gain as much as a perfect detector (\cref{fig:breakeven}, \cref{tab:headroom}).

\begin{figure}[t]
\centering
\includegraphics{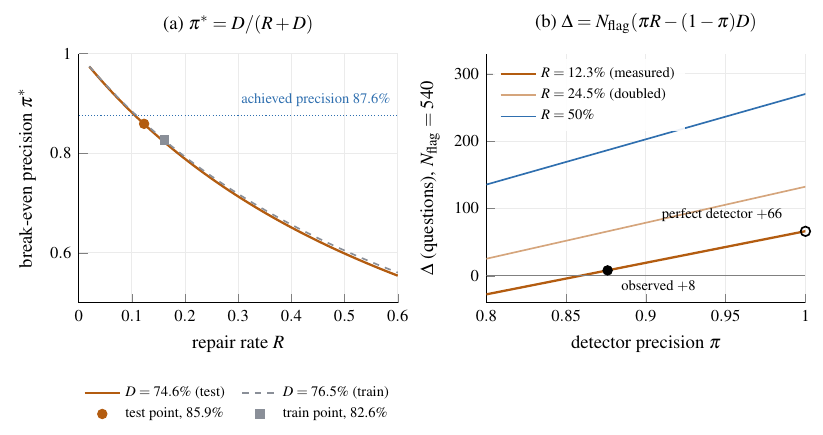}
\caption{The break-even condition. (a) Precision a detector needs before retries help, as a function of the repair rate, at the measured damage rates. (b) Net gain on the 540-flag test budget for three repair rates. Moving from the observed operating point to perfect precision gains 58 questions; doubling $R$ at the current precision gains about the same.}
\label{fig:breakeven}
\end{figure}

\begin{table}[t]
\centering
\caption{Counterfactual gains on the 540-flag test budget from \eqref{eq:delta}, with all other quantities at their measured values. These are first-order estimates, not experiments.}
\label{tab:headroom}
\small
\begin{tabular}{@{}lrr@{}}
\toprule
change & $\Delta$ (questions) & accuracy (points) \\
\midrule
none (observed) & $+8$ & $+0.6$ \\
no damage ($D=0$) & $+58$ & $+4.4$ \\
perfect detector ($\pi=1$) & $+66$ & $+5.0$ \\
doubled repair rate ($R=24.5\%$) & $+66$ & $+5.0$ \\
repair rate $R=50\%$ & $+187$ & $+14.1$ \\
perfect detector, unlimited budget & $+117$ & $+8.9$ \\
\bottomrule
\end{tabular}
\end{table}

\subsection{Why is the repair rate so low?}
\label{sec:why-r}

\paragraph{Empirical analysis.}
The low repair rate is not a matter of a careless retry.
Paired on the 1,068 training failures, the diagnosis-conditioned retry recovers 18.5\% against 15.1\% for an untyped retry (exact McNemar $p=0.010$); it is the best repair we built, and on test it still reaches only 12.3\%.
Only 32.9\% of failures contain a locatable false equation, and in 51.6\% of those the false equation is the first step.
At most $0.329\times(1-0.516)\approx16\%$ of failures therefore contain a correct prefix before their first error (\cref{fig:composition}).

\begin{figure}[t]
\centering
\includegraphics{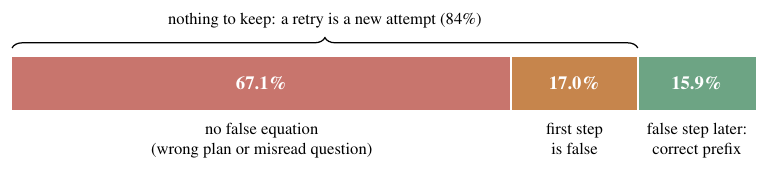}
\caption{Where the first error lies in pass-1 failures. Only failures whose first false equation comes after the first step (green) contain a correct prefix that a repair could preserve. The 17.0\% and 15.9\% split the 32.9\% of failures with a locatable false equation by whether that equation is the first step (51.6\%).}
\label{fig:composition}
\end{figure}

\paragraph{Intuitive explanation.}
Repair in the ordinary sense keeps a mostly correct derivation and fixes a localised fault.
For the remaining 84\% of failures there is nothing to keep, either because the trace contains no false equation at all or because it goes wrong at the first step, and a repair becomes a new attempt at the whole problem.
That attempt is also made on exactly the questions the model has just failed, where its accuracy is far below its 27.8\% average.
For a model whose failures are rarely localised, repairing a trace and producing a correct trace from scratch are the same problem; a verifier, however precise, can only route questions to that problem and cannot make it easier.

\section{Conclusion and Discussion}
\label{sec:conclusion}

Our work shows that, for the model we study, checking a trace against itself cannot detect most failures: two thirds of them contain no false equation, so the internal-consistency signals we evaluated reach a ceiling that a better classifier over the same signals is unlikely to lift.
Testing the trace as a program hypothesis against the model's own behaviour reaches those failures without a reference solution, a judge or a template, yet better detection does not improve accuracy, because a retry rarely repairs what it is sent to fix.
We draw three recommendations from these results.

\paragraph{Reporting repair and damage rates alongside detectors.}
AUROC alone does not determine whether a verifier improves outcomes.
\Cref{eq:breakeven} shows that a detector helps only when its precision exceeds a threshold set entirely by the action it triggers.
We encourage work on reference-free verification to report the repair and damage rates of that action, whether it is a retry, an escalation to a larger model, a tool call or a human review.

\paragraph{Testing a trace against the model, not only against itself.}
The two verifiers fail in complementary places, and in an exploratory analysis their combination detected failures better than either alone.
Checks of a trace against itself catch local errors in the arithmetic; tests of the trace against the model's behaviour under intervention catch errors in the choice of arithmetic.
Because support is strong evidence and refutation is weak, such tests are best used to raise precision, by confirming or vetoing a candidate, rather than to trigger retries.

\paragraph{Measuring failure composition before building repair.}
Whether detection or repair is the bottleneck depends on how a model fails.
A model whose failures are mostly localised slips sits in a very different region of \cref{fig:breakeven} from the one studied here.
Measuring the share of failures with no false equation, and the position of the first error, is cheap, needs no retry machinery, and indicates in advance whether a better detector can help.

\section{Limitations}
\label{sec:limitations}

Although our measurements are fully instrumented, they come from one model, one dataset and one prompting setup, and LLaMA-2-7B-chat, at 27.8\% pass-1 accuracy, is weak by current standards.
The accounting identity, the argument that self-checking cannot see errors of plan, and the mechanism by which a candidate-derived relation exposes them do not depend on the model; the share of failures without a false equation, the repair and damage rates, and therefore the conclusion that repair is the bottleneck may not transfer.
The probe refutes 45.5\% of correct traces, which follows from the model's 28.6\% accuracy on perturbed questions; we expect fewer false refutations from a stronger model but have not measured this.
Our parser does not distinguish asserted equations from hypothetical or retracted ones (``suppose $x=5$'', ``$12\times4=46$ \ldots{} no, 48''): a retracted false equation lowers equation validity, while the abstraction, which binds later steps to earlier results by value, turns a retracted result that is never reused into dead code. We did not measure how often these patterns occur, and models that self-correct frequently would require a retraction-aware parser.
The probe applies to 71\% of test questions and requires numeric quantities that the trace uses, so it does not reach proof-style reasoning.
Finally, we compare the probe with the static detector, but not empirically with self-consistency, an LLM judge or a trained process verifier; \cref{tab:compare} compares properties only.

\section*{Reproducibility Statement}

All experiments use the publicly available LLaMA-2-7B-chat with frozen weights and the public GSM8K benchmark.
Every measurement is a deterministic function of the generated traces, computed with regular expressions and sympy.
\Cref{app:signals,app:taxonomy,app:guards,app:calibration,app:arms} give the signal definitions, failure taxonomy, probe guards, detector calibration and all policy results, and \cref{app:audit} lists every error we found and corrected during the project together with its effect on the reported numbers.
The 988 probes behind the test results cost about 4.3 GPU-hours on a single NVIDIA T4.

\paragraph{Code availability.} The code, the notebooks used to run each stage, and the stored test-set outputs needed to recompute the reported test-set numbers on a CPU are available at \url{https://github.com/S2V3/Detection-is-not-enough}. The raw outputs of every run are archived at \url{https://github.com/S2V3/fade/tree/results}.


\clearpage
\appendix

\section{Trace Signals}
\label{app:signals}

All four signals are computed from the equation chain and the question, and none reads a reference answer.
\emph{Equation validity} is the fraction of stated equations that the algebra system confirms; a trace that writes $48-3=44$ loses here even when its plan is right.
\emph{Value overlap} compares the numbers the question gives with those the trace uses; a trace that ignores a given quantity, or introduces one that appears nowhere in the question, separates from a correct trace on this signal, and percentages are first normalised to a single convention.
\emph{Goal engagement} asks whether the trace ever computes the quantity the question asks for, matched on units and on the wording of the final question sentence, which catches a trace that solves a related quantity and stops.
\emph{Coherence} checks that every step refers to a value the question gave or an earlier step produced, rather than a number from nowhere.
A fifth signal, \emph{redundancy}, was built and then dropped: it showed no separation between correct and wrong traces, and at its threshold it flagged legitimately parallel steps as duplicates.
It is logged as a diagnostic and never enters a decision.

\section{Failure Taxonomy}
\label{app:taxonomy}

The diagnosis is symbolic.
A cascade of rules reads the equation chain, the question's numbers and the trace's structure, and abstains when no rule fires with enough evidence.
\Cref{tab:taxonomy} lists the eight types.
Structural failures have clear signatures and semantic ones do not: \type{SM} and \type{CE} are recognised with low recall (38.2\% and 42.5\%) and are two of the three largest classes, so more than half of those failures receive a mismatched diagnosis.
This caps the benefit of typed repair, and it is why the deployed pipeline falls back to an untyped retry below a confidence gate.

\begin{table}[h]
\centering
\caption{Failure types over the 1,068 training-split pass-1 failures, with the type head's out-of-fold recall (GroupKFold by question, 2,568 traces) and the count of each type among the 579 wrong testable test traces of \cref{fig:per-type}.}
\label{tab:taxonomy}
\small
\begin{tabular}{@{}llrrrr@{}}
\toprule
type & what went wrong & train $n$ & share & recall & test $n$ \\
\midrule
\type{WP} & wrong plan, sound arithmetic & 254 & 23.8\% & 69.2\% & 180 \\
\type{SM} & misread the question, solved a different problem & 253 & 23.7\% & 38.2\% & 185 \\
\type{CE} & computation error inside a good plan & 136 & 12.7\% & 42.5\% & 83 \\
\type{ST} & a needed intermediate never computed & 114 & 10.7\% & 62.8\% & 56 \\
\type{AL} & sets up symbols, never solves them & 94 & 8.8\% & 87.1\% & 28 \\
\type{UNCL} & no rule fired confidently & 88 & 8.2\% & 51.7\% & 35 \\
\type{NR} & never engaged with the question & 73 & 6.8\% & 95.3\% & 0 \\
\type{TR} & truncated generation & 56 & 5.2\% & 98.9\% & 12 \\
\bottomrule
\end{tabular}
\end{table}

\section{Probe Guards}
\label{app:guards}

A probe is generated only when all five of the following conditions hold.
\begin{enumerate}[leftmargin=*]
  \item \textbf{Load-bearing.} Some candidate's chain consumes the perturbed number. Changing an unused number should not move the answer, so good and bad traces would appear equally consistent.
  \item \textbf{Unique occurrence.} The number occurs exactly once in the question, so the rewrite is unambiguous.
  \item \textbf{No collision.} The replacement does not equal another number in the question, which could make the abstraction bind the wrong slot.
  \item \textbf{Sensible question.} Integrality and sign survive the change, so the rewritten question remains a valid word problem.
  \item \textbf{Discriminative.} At least one candidate's prediction moves under the substitution; otherwise the probe cannot distinguish candidates whatever the model answers.
\end{enumerate}
Each probe takes 15.8\,s on a T4.
Probe order is deterministic but depends on the target-preference rules, so resumption is keyed on the perturbation itself, the pair of original and replacement values, rather than on a probe index.
This also allows a one-probe run to be extended to two without regenerating anything; probe 0 is chosen identically either way, which we verified on 400 test rows with no mismatches.

\section{Detector Calibration}
\label{app:calibration}

A budget quantile needs the whole score distribution, which a pipeline answering one question at a time does not have.
The mapping from budget to threshold is therefore measured out of fold and stored with the model, so that each question can be judged on its own.

\begin{table}[h]
\centering
\caption{Budget-to-threshold calibration, measured out of fold on the 1,500 training-split pass-1 traces, as recorded in the calibration log. Expected net is $\mathrm{TP}\cdot R-\mathrm{FP}\cdot D$ with the training-run rates $R=16.1\%$ and $D=76.5\%$.}
\label{tab:calibration}
\small
\begin{tabular}{@{}rrrrr@{}}
\toprule
budget & flagged & $\pwrong$ cut & precision & expected net \\
\midrule
10\% & 150 & 0.90 & 97.3\% & $+20.4$ \\
20\% & 300 & 0.83 & 95.7\% & $+36.3$ \\
40\% & 600 & 0.683 & 90.7\% & $+44.7$ \\
50\% & 750 & 0.60 & 88.0\% & $+37.4$ \\
70\% & 1,050 & 0.48 & 82.1\% & $-5.0$ \\
100\% & 1,500 & 0.13 & 71.2\% & $-158.5$ \\
\bottomrule
\end{tabular}
\end{table}

Each row's expected net is computed from the flagged set at that budget as $\mathrm{TP}\cdot R-\mathrm{FP}\cdot D$, where TP and FP are the numbers of wrong and correct traces flagged and $R=16.1\%$ and $D=76.5\%$ are the repair and damage rates from the training run that the calibration uses.
For example, the 40\% budget flags 600 traces, of which 544 are wrong (90.7\%), giving $544\times0.161-56\times0.765=+44.7$; every row of \cref{tab:calibration} is reproduced in the same way.
The resulting break-even precision, $D/(R+D)=82.6\%$, falls between the 50\% and 70\% budgets, and the expected net is largest at 40\%.
The test pipeline uses the 40\% cut of 0.683.
\Cref{fig:calibration} plots the calibration against both break-even points: on the training split the detector stays profitable up to a budget between 50\% and 70\%, whereas on test the lower repair rate moves the break-even up to 85.9\%, which the test operating point clears by a small margin.

\begin{figure}[h]
\centering
\includegraphics{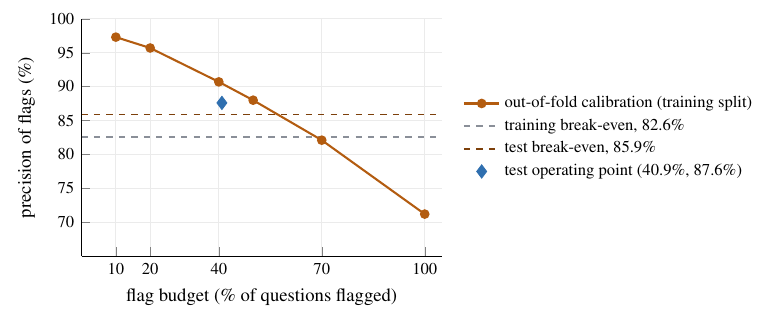}
\caption{Precision of the static detector's flags as a function of the flag budget. Below its break-even line, retrying flagged answers loses more correct answers than it repairs.}
\label{fig:calibration}
\end{figure}

\section{All Policies and Per-Class Tests}
\label{app:arms}

The \emph{veto} policy cancels a retry when the probe supports the first trace.
It should help in principle, since removing false positives raises precision, but it rescued only 3 of the 67 false positives, because flagged traces are overwhelmingly wrong and wrong traces are rarely supported.
\emph{Back-substitution} adds the model's answer to $Q'$, mapped back to the original numbers through the program, as an extra candidate at no extra cost; it changed no decision.
Fisher exact tests of each class's refutation rate against the 45.5\% rate on correct traces (122 of 268) give
\type{WP} 140/180, $p=5.5\times10^{-12}$;
\type{SM} 157/185, $p=3.3\times10^{-18}$;
\type{CE} 69/83, $p=6.6\times10^{-10}$;
\type{ST} 47/56, $p=1.1\times10^{-7}$;
\type{AL} 27/28, $p=6.5\times10^{-8}$;
\type{UNCL} 33/35, $p=8.6\times10^{-9}$;
\type{TR} 11/12, $p=1.9\times10^{-3}$.
These sum to 484 refuted wrong traces, which together with the 122 refuted correct traces give the 606 refutations in \cref{tab:support}.

\section{Errors Found and Corrected}
\label{app:audit}

We list every error we found that could have affected a reported number, with its effect and its correction.
The numbers in this appendix were measured on the intermediate version of the pipeline in which each problem was found, before the later corrections had been made.
They quantify the size of each correction at the time it was made and are not directly comparable with one another or with \cref{tab:detector}, which reports the final detector evaluated with every correction below in place.

\paragraph{A reference-derived feature reached the classifier.}
The number of annotated solution steps comes from the reference solution and had entered the feature set, lifting out-of-fold detection AUROC in that version from 0.717 to 0.791.
It cannot be computed where no annotation exists, so we removed it and added an assertion that rejects such features.

\paragraph{Question-level leakage inflated the type head.}
Each question contributes a first attempt and a retry.
StratifiedKFold split on traces, so a question's first attempt could sit in training while its retry sat in test, letting the model recognise the question instead of reading the trace.
In the version where this was found, switching to GroupKFold left the detection head essentially unchanged (AUROC 0.752 to 0.753) and lowered the type head by 4.4 points (0.657 to 0.613). These figures are computed over all 2,568 training traces, first attempts and retries together; \cref{tab:detector} reports first attempts only, the population the deployed detector sees, and the two populations differ in how legible their errors are.
The cause was a question-level TF--IDF feature, and all reported numbers use GroupKFold.

\paragraph{Union scoring is not deployable.}
Counting a question as solved when either the first answer or the retry is correct requires a reference answer to choose between them, and a deployed system that replaced its first answer cannot fall back on the answer it rejected.
Every accuracy in this paper is what a deployed system would return.

\paragraph{The blend weight.}
An earlier version chose the detection blend weight by maximising AUROC on the test set, which gave 0.778 at a weight of 0.25 on the probe score. We now report the equal-weight value, 0.773, together with the range over weights from 0.1 to 0.9 (0.769 to 0.778). Because the weight grid was first examined on the test set, we treat the blended result as exploratory.

\paragraph{Typed versus untyped repair at test.}
On the training split the diagnosis comes from the symbolic cascade and is essentially always right, while at test it comes from the classifier, which is right 59\% of the time, and 28\% of retries fall back to the untyped arm.
The test-time comparison is therefore underpowered rather than negative: with 94 paired retries, a true five-point difference would be detected 14\% of the time and an eight-point difference 33\% of the time.

\paragraph{Extraction.}
Two parser defects were found and corrected.
Percentages were emitted in both conventions, which inflated value overlap for any trace that mentioned one.
Binary minus was read as a sign, so ``60-20'' was tokenised as two numbers and the equation silently dropped; this affected 14 lines across 13 traces (0.9\% of the training run), all of them subtractions, which is why it was hard to notice.


\begin{thebibliography}{30}
\small
\setlength{\itemsep}{1pt}

\bibitem[Bottman et~al.(2026)]{bottman2026operadic}
N.~Bottman, Y.~Liu, and K.~Richardson.
\newblock Operadic consistency: a label-free signal for compositional reasoning failures in {LLMs}.
\newblock \emph{arXiv:2606.13649}, 2026.

\bibitem[Chen et~al.(1998)]{chen1998metamorphic}
T.~Y. Chen, S.~C. Cheung, and S.~M. Yiu.
\newblock Metamorphic testing: a new approach for generating next test cases.
\newblock Technical Report HKUST-CS98-01, Hong Kong University of Science and Technology, 1998.

\bibitem[Chen et~al.(2018)]{chen2018mtreview}
T.~Y. Chen, F.-C. Kuo, H.~Liu, P.-L. Poon, D.~Towey, T.~H. Tse, and Z.~Q. Zhou.
\newblock Metamorphic testing: a review of challenges and opportunities.
\newblock \emph{ACM Computing Surveys}, 51(1), 2018.

\bibitem[Chen et~al.(2023)]{chen2023pot}
W.~Chen, X.~Ma, X.~Wang, and W.~W. Cohen.
\newblock Program of thoughts prompting: disentangling computation from reasoning for numerical reasoning tasks.
\newblock \emph{Transactions on Machine Learning Research}, 2023.

\bibitem[Cobbe et~al.(2021)]{cobbe2021gsm8k}
K.~Cobbe, V.~Kosaraju, M.~Bavarian, M.~Chen, H.~Jun, L.~Kaiser, M.~Plappert, J.~Tworek, J.~Hilton, R.~Nakano, C.~Hesse, and J.~Schulman.
\newblock Training verifiers to solve math word problems.
\newblock \emph{arXiv:2110.14168}, 2021.

\bibitem[Gao et~al.(2023)]{gao2023pal}
L.~Gao, A.~Madaan, S.~Zhou, U.~Alon, P.~Liu, Y.~Yang, J.~Callan, and G.~Neubig.
\newblock {PAL}: Program-aided language models.
\newblock In \emph{International Conference on Machine Learning}, 2023.

\bibitem[Huang et~al.(2024)]{huang2024selfcorrect}
J.~Huang, X.~Chen, S.~Mishra, H.~S. Zheng, A.~W. Yu, X.~Song, and D.~Zhou.
\newblock Large language models cannot self-correct reasoning yet.
\newblock In \emph{International Conference on Learning Representations}, 2024.

\bibitem[Kamoi et~al.(2024)]{kamoi2024survey}
R.~Kamoi, Y.~Zhang, N.~Zhang, J.~Han, and R.~Zhang.
\newblock When can {LLMs} actually correct their own mistakes? {A} critical survey of self-correction of {LLMs}.
\newblock \emph{Transactions of the Association for Computational Linguistics}, 12, 2024.

\bibitem[Lanham et~al.(2023)]{lanham2023faithfulness}
T.~Lanham, A.~Chen, A.~Radhakrishnan, B.~Steiner, C.~Denison, D.~Hernandez, D.~Li, E.~Durmus, E.~Hubinger, J.~Kernion, et~al.
\newblock Measuring faithfulness in chain-of-thought reasoning.
\newblock \emph{arXiv:2307.13702}, 2023.

\bibitem[Lightman et~al.(2023)]{lightman2023verify}
H.~Lightman, V.~Kosaraju, Y.~Burda, H.~Edwards, B.~Baker, T.~Lee, J.~Leike, J.~Schulman, I.~Sutskever, and K.~Cobbe.
\newblock Let's verify step by step.
\newblock \emph{arXiv:2305.20050}, 2023.

\bibitem[Lyu et~al.(2023)]{lyu2023faithful}
Q.~Lyu, S.~Havaldar, A.~Stein, L.~Zhang, D.~Rao, E.~Wong, M.~Apidianaki, and C.~Callison-Burch.
\newblock Faithful chain-of-thought reasoning.
\newblock In \emph{Proceedings of IJCNLP-AACL}, 2023.

\bibitem[Madaan et~al.(2023)]{madaan2023selfrefine}
A.~Madaan, N.~Tandon, P.~Gupta, S.~Hallinan, L.~Gao, S.~Wiegreffe, U.~Alon, N.~Dziri, S.~Prabhumoye, Y.~Yang, et~al.
\newblock Self-Refine: iterative refinement with self-feedback.
\newblock In \emph{Advances in Neural Information Processing Systems}, 2023.

\bibitem[Mirzadeh et~al.(2024)]{mirzadeh2024gsmsymbolic}
I.~Mirzadeh, K.~Alizadeh, H.~Shahrokhi, O.~Tuzel, S.~Bengio, and M.~Farajtabar.
\newblock {GSM-Symbolic}: Understanding the limitations of mathematical reasoning in large language models.
\newblock \emph{arXiv:2410.05229}, 2024.

\bibitem[Ribeiro et~al.(2020)]{ribeiro2020checklist}
M.~T. Ribeiro, T.~Wu, C.~Guestrin, and S.~Singh.
\newblock Beyond accuracy: behavioral testing of {NLP} models with {CheckList}.
\newblock In \emph{Proceedings of ACL}, 2020.

\bibitem[Sahoo et~al.(2025)]{sahoo2025catch}
S.~Sahoo, V.~Jain, S.~Vats, S.~Mohapatra, R.~Min, A.~Chadha, and D.~Chaudhary.
\newblock Catch me if you can: How smaller reasoning models pretend to reason with mathematical fidelity.
\newblock \emph{arXiv:2512.00552}, 2025.

\bibitem[Segura et~al.(2016)]{segura2016survey}
S.~Segura, G.~Fraser, A.~B. Sanchez, and A.~Ruiz-Cort\'es.
\newblock A survey on metamorphic testing.
\newblock \emph{IEEE Transactions on Software Engineering}, 42(9), 2016.

\bibitem[Shukla et~al.(2026)]{shukla2026silent}
V.~Shukla, V.~Shukla, Atul, D.~Mishra, and M.~K. Das.
\newblock A reference-free score for detecting silent reasoning failures in large language models.
\newblock \emph{arXiv:2607.26102}, 2026.

\bibitem[Touvron et~al.(2023)]{touvron2023llama2}
H.~Touvron, L.~Martin, K.~Stone, P.~Albert, A.~Almahairi, Y.~Babaei, N.~Bashlykov, S.~Batra, P.~Bhargava, S.~Bhosale, et~al.
\newblock Llama 2: open foundation and fine-tuned chat models.
\newblock \emph{arXiv:2307.09288}, 2023.

\bibitem[Turpin et~al.(2023)]{turpin2023unfaithful}
M.~Turpin, J.~Michael, E.~Perez, and S.~R. Bowman.
\newblock Language models don't always say what they think: unfaithful explanations in chain-of-thought prompting.
\newblock In \emph{Advances in Neural Information Processing Systems}, 2023.

\bibitem[Uesato et~al.(2022)]{uesato2022process}
J.~Uesato, N.~Kushman, R.~Kumar, F.~Song, N.~Siegel, L.~Wang, A.~Creswell, G.~Irving, and I.~Higgins.
\newblock Solving math word problems with process- and outcome-based feedback.
\newblock \emph{arXiv:2211.14275}, 2022.

\bibitem[Wang et~al.(2023)]{wang2023selfconsistency}
X.~Wang, J.~Wei, D.~Schuurmans, Q.~Le, E.~Chi, S.~Narang, A.~Chowdhery, and D.~Zhou.
\newblock Self-consistency improves chain of thought reasoning in language models.
\newblock In \emph{International Conference on Learning Representations}, 2023.

\bibitem[Wei et~al.(2022)]{wei2022cot}
J.~Wei, X.~Wang, D.~Schuurmans, M.~Bosma, B.~Ichter, F.~Xia, E.~Chi, Q.~Le, and D.~Zhou.
\newblock Chain-of-thought prompting elicits reasoning in large language models.
\newblock In \emph{Advances in Neural Information Processing Systems}, 2022.

\bibitem[Weng et~al.(2023)]{weng2023selfverification}
Y.~Weng, M.~Zhu, F.~Xia, B.~Li, S.~He, S.~Liu, B.~Sun, K.~Liu, and J.~Zhao.
\newblock Large language models are better reasoners with self-verification.
\newblock In \emph{Findings of the Association for Computational Linguistics: EMNLP 2023}, 2023.

\bibitem[Zhang et~al.(2024)]{zhang2024small}
Y.~Zhang, M.~Khalifa, L.~Logeswaran, J.~Kim, M.~Lee, H.~Lee, and L.~Wang.
\newblock Small language models need strong verifiers to self-correct reasoning.
\newblock In \emph{Findings of the Association for Computational Linguistics: ACL 2024}, 2024.

\bibitem[Zhong et~al.(2026)]{zhong2026verytrace}
N.~Zhong, A.~E. Tanriverdi, K.~Kale, and S.~Vishwanath.
\newblock {VeryTrace}: Verifying reasoning traces through compilable formalism and structured verification.
\newblock In \emph{LM4Plan Workshop at ICML}, 2026.

\bibitem[Zhou et~al.(2025)]{zhou2025checklist}
Z.~Zhou, S.~Liu, M.~Ning, W.~Liu, J.~Wang, D.~F. Wong, X.~Huang, Q.~Wang, and K.~Huang.
\newblock Is your model really a good math reasoner? {E}valuating mathematical reasoning with checklist.
\newblock In \emph{International Conference on Learning Representations}, 2025.

\bibitem[Zhou et~al.(2026)]{zhou2026lgmt}
Z.~Zhou, M.~Li, X.~Fang, X.~Zhou, W.~Lin, and Z.~Zheng.
\newblock {LGMT}: Logic-grounded metamorphic testing for evaluating the reasoning reliability of {LLMs}.
\newblock \emph{arXiv:2605.23965}, 2026.

\end{thebibliography}
\end{document}